\documentclass[11pt]{article}
\usepackage{color}
\usepackage{epsfig}
\usepackage{amsmath}
\usepackage{amsfonts}
\usepackage{amssymb}
\usepackage{amsthm}
\usepackage{graphicx}
\usepackage[pdfborder={0 0 0},linktocpage=true]{hyperref}

\newtheorem{lemma}{Lemma}
\newtheorem{theorem}{Theorem}

\newtheorem{corol}{Corollary}

\newtheorem{condition}{Condition}

\newcommand{\Real}{{\mathbb{R}}^}
\newcommand{\Reals}{\mathbb{R}}
\newcommand{\tgi}{\lim_{t\rightarrow \infty}}

\newcommand{\mE}{{\mathcal E}_}

\newcommand{\diam}{{\rm diam}}
\newcommand{\dist}{{\rm dist}}

\title{\LARGE Linear Matrix Inequalities for Ultimate Boundedness of Dynamical Systems with Conic Uncertain/Nonlinear Terms}
\author{\Large Beh\c cet\ A\c c\i kme\c se \\
\small Department of Aerospace Engineering and Engineering Mechanics \\
\small The University of Texas at Austin, USA}

\begin{document}
\maketitle
\begin{abstract}
This note introduces a sufficient   Linear Matrix Inequality (LMI)  condition for the ultimate boundedness of a class of continuous-time dynamical systems with conic uncertain/nonlinear terms.  
\end{abstract}

\section{Introduction}
This note introduces an LMI \cite{LMIbalakbk} result for the ultimate boundedness of dynamical systems with conic uncertain/nonlinear terms. Earlier research developed necessary and sufficient conditions for quadratic stability for systems with similar characterizations for uncertainties and nonlinearities \cite{behcet_qs_scl,balak_cdc94,balakash,yakup_kyp,rantzer_kyp,sweithesis}.  Incremental version of these characterizations are used  in the synthesis of nonlinear observers \cite{arcak01,arcakobs_aut_01,arcakCobs_scl_01,mythesis,dani2012observer,behcet_05observer,behcet_11observers} and to design robust Model Predictive Control (MPC) algorithms \cite{accikmese2006nonlinear,carson2006model,behcet_mpc11,carson2013robust}. The following results first appeared in \cite{behcet_mpc11}.

\paragraph*{Notation:}
%=====================
%
The following is a partial list of notation used in this paper:
$Q\!=\!Q^T >(\geq) 0$ implies $Q$ is a positive-(semi-)definite matrix; $\,
Co\{G_1,\ldots,G_N\}$ represents the convex hull of matrices $G_1
,\ldots, G_N$; 
%$\Real n$ is the space of $n$ dimensional vectors
%with real components; $\Reals$ is the set of real numbers; 
%$\Real n$ is the $n$ dimensional real vector space;
 $\mathbb{Z}^{+}$ is the set of non-negative
integers; $\|v\|$ is the 2-norm of the vector $v$;\; ${\lambda}_{max} (P)$ and ${\lambda}_{min}(P)$ are maximum and minimum eigenvalues of symmetric matrix $P$;  $\mE P \!:=\! \{ x : \, x^T P x \! \leq \! 1\}$ is an ellipsoid (possibly not bounded) defined by $P=P^T \geq 0$; for a bounded signal $w(\cdot)$, $\| w\|_{[t_1,t_2]} := \sup_{\tau \in [t_1,t_2]} \|w(\tau)\| $; 
%a function is said to be of class  $\mathcal{C}^r$ if its derivatives up to order $r$ exist and continuous; 
for a compact set $\Omega$,  $\diam (\Omega) := \max_{x,y \in \Omega} \| x-y\|$ and $\dist(a,\Omega) := \min_{x \in \Omega} \| a-x \| $; and, for $V: \Real N \! \to\! \Reals$, $\nabla V = [\,
\partial V / \partial {x_1} \, \ldots \, \partial V / \partial {x_n}\,]$.
A set $\Omega$ is said to be {\em invariant} over $[t_0,\infty)$ for $\dot{x}=f(x,t)$ if: $x(t_0) \in \Omega$ implies that $x(t) \in \Omega$, $\forall \, t \geq t_0$. $\Omega$ is also {\em attractive} if for every $x(t_0)$, ${\displaystyle \tgi \dist(x(t),\Omega)=0}$. 

\section{A General Analysis Result on Ultimate Boundedness} 
The following lemma gives a Lyapunov characterization for the ultimate boundedness of a nonlinear time-varying system, which is used in the proof of main result.
%--
\begin{lemma} \label{lm:iae}
Consider a  system  with state $\eta$ and  input
 $\sigma$ described by
\begin{equation} \label{eq:iosyshalf}
\dot{\eta} = \phi(t,\eta,\sigma)\,, \qquad  t \geq t_0.
\end{equation}
 Suppose there exists  a positive definite  symmetric matrix
 $P$ with, $V(\eta) = \eta^T P \eta$, and a continuous function $W$ such that
 for all $\eta$, $\sigma$ and $t \geq 0$
\begin{equation} \label{eq:suffiae}
\dot{V} =  2 \eta ^T P \phi(t,\eta,\sigma) \leq - W(\eta) < 0 \qquad \mbox{when} \qquad \eta^T
P \eta > \| \sigma \|^2 \, .
\end{equation}
 Then for every bounded continuous input signal $\,
\sigma(\cdot) \,$, the ellipsoid
 $
{\cal E} :=
\{ \eta
: \eta^T P \eta \leq \|\sigma(\cdot) \|_{[t_0,\infty)}^2 \} 
$
 is invariant and attractive  for   system
(\ref{eq:iosyshalf}).
Furthermore, for any solution $\eta(\cdot)$  we have
\begin{equation} \label{eq:limsupV}
\limsup_{t \rightarrow \infty} \left[ \eta(t)^T P \eta(t) \right] \leq
\|\sigma(\cdot) \|_{[t_0,\infty)}^2 \, .
\end{equation}
\end{lemma}
See \cite{behcet_ijc03} for a proof of the above lemma.
%---

%---
\section{Analysis of Systems with Conic Uncertainty/\\Nonlinearity}
In this section we consider the following system
\begin{equation} \label{eq:Conesys}
\dot{x} = A x + E p(t,x) + G w 
\end{equation}
where $x$ is the state, $p$ represents the uncertain/nonlinear terms, and $w$ is a bounded disturbance signal, and $p \in \mathcal{F}(\mathcal{M})$ with
\begin{equation} \label{eq:Coneq}
q = C_q x + D p.
\end{equation} 
To define $p \in \mathcal{F}(\mathcal{M})$,  let
\begin{equation} \label{eq:Fset}
	\mathcal{F}  (\mathcal{M}):= \left\{ \phi: \Real {n_q\!+\!1} \! \to \! \Real {n_p}  : \  \mbox{$ \phi$ satisfies QI (\ref{eq:QI})} \right\}.
\end{equation}
where 
the following QI ({\em Quadratic Inequality})  is satisfied
\begin{equation} \label{eq:QI}
	\left[\!\! \begin{array}{c}
	q \\
	\phi(t,v)
	\end{array}\!\! \right]^T  M \left[\! \! \begin{array}{c}
	q \\
	\phi(t,v)
	\end{array}\!\! \right] \geq 0,  \qquad \forall \, M \in \mathcal{M}, \ \ \forall \, v \in \Real {n_q},  \ \mbox{and $ \forall \, t.$}
\end{equation}
where $\mathcal{M}$ is  a set of symmetric matrices.

The following condition, which is instrumental in the  control synthesis, is assumed to hold for the incrementally-conic uncertain/nonlinear terms. 
\begin{condition} \label{asm:Multip4syn}
	There exist a nonsingular matrix $T$ and a convex set $\mathcal{N}$ of matrix pairs $(X,Y)$ with $Y \in \Real {n_p \times n_p}$ and $X,Y$
	symmetric and nonsingular such that for each $(X,Y) \in \mathcal{N}$, the matrix
	\begin{equation} \label{eq:N4M}
		M=T^T \left[\begin{array}{cc}X^{-1} & 0 \\0 & -Y^{-1}\end{array}\right] T\in \mathcal{M} \quad \ \text{\rm with} \ \quad T =
		\left[\begin{array}{cc}T_{11} & T_{12} \\ T_{21} & T_{22}\end{array}\right],
	\end{equation}
	where $T_{22}+T_{21} D$ is nonsingular, $T_{21} \!\in\! \Real {n_p\times n_q}$ and $T_{22} \in \Real {n_p \times n_p}$. Furthermore,  the set $\mathcal{N}$ can be parameterized by a finite number of LMIs.
\end{condition}

It is also assumed that the set of multipliers $\mathcal{M}$ satisfies Condition \ref{asm:Multip4syn}. The following theorem, the main result of this note,  presents an LMI condition guaranteeing ultimate boundedness of all the trajectories of the system  (\ref{eq:Conesys}). 
%---
\begin{theorem} \label{lem:UltBoundCone}
	Consider the system given by (\ref{eq:Conesys}) with $p \in \mathcal{F}(\mathcal{M})$ where the multiplier set $\mathcal{M}$ satisfies
	Condition \ref{asm:Multip4syn}.  Suppose that there exist $Q=Q^T>0$, $(X,Y) \in \mathcal{N}$,  $\lambda>0$, and $R=R^T>0$ such that
	the following matrix inequality holds
	{\small
	\begin{equation} \label{eq:AnalysisLMI}
	\left[\begin{array}{cccc}
	(A\!-\! E \Gamma^{-1} T_{21} C_q)Q +Q(A\!-\! E \Gamma^{-1} T_{21} C_q)^T+\! \lambda Q+ \! R  & E \, \Gamma^{-1} Y &Q \, C_q^T \Sigma^T   & G \\
	Y \Gamma^{-T} E^T  & -Y & Y \Lambda^T &  0 \\
	\Sigma \, C_q \, Q  & \Lambda  Y& -X  & 0 \\
	G^T & 0 & 0 & -\lambda I \end{array}\right] \leq 0
	\end{equation}}
	where
	$$
	\Gamma = T_{21} D + T_{22}, \ \ \Lambda = (T_{11} D + T_{12}) \Gamma^{-1}, \ \
	\Sigma = T_{11}- (T_{11} D + T_{12} ) \Gamma^{-1} T_{21}.
	$$
	Then, letting $V(x):=x^{T} Q^{-1}x$, we have
	\begin{equation} \label{eq:VdotVgw}
	\dot{V} (x) + x^{T}Q^{-1} R Q^{-1} x  \leq 0 ,\qquad \forall \, V(x)\geq \|w\|^{2}.
	\end{equation}
\end{theorem}
%---
\begin{proof}
	%
	%We  show that the system (\ref{eq:Conesys})   satisfies the condition (\ref{eq:suffiae}) with $V(x) = x^T Q^{-1} x$. 
	First pre- and post-multiply (\ref{eq:AnalysisLMI}) by
	\begin{equation*}
		\left[\begin{array}{cccc} I & 0 &0 &0 \\ 0 & I & 0 & 0\\
		0 & 0 & 0 & I\\0 & 0 & I & 0\end{array}\right] 
	\end{equation*}
	and then pre- and post-multiply the resulting matrix inequality with ${
	\rm diag} (Q^{-1},Y^{-1},I,I)$ to obtain
	{\small
	\begin{equation*}
		\left[\begin{array}{cccc}
		%(A\!-\! E \Gamma^{-1} T_{21} C_q)Q +Q(A\!-\! E \Gamma^{-1} T_{21} C_q)^T+\! \lambda Q+ \! R  & Q^{-1} E \, \Gamma^{-1}  
		\left( \begin{array}{c} Q^{-1} (A\!-\! E \Gamma^{-1} T_{21} C_q) +(A\!-\! E \Gamma^{-1} T_{21} C_q)^T Q^{-1} \\ +\, \lambda Q^{-1} + Q^{-1} RQ^{-1}
		\end{array}\right)& Q^{-1} E \Gamma^{-1}& Q^{-1} G    & C_q^T \Sigma^T \\
		\Gamma^{-T} E^T Q^{-1} & -Y^{-1} &  0&  \Lambda^T \\
		G^T Q^{-1}  & 0  & -\lambda I  & 0 \\
		\Sigma \, C_q  & \Lambda & 0 & -X \end{array}\right] \leq 0
	\end{equation*}
	}
	By using Schur complements the above inequality implies that %to show that (\ref{eq:AnalysisLMI}) implies the following inequality
	{\small
	\begin{equation*}
	\begin{split}
	&\left[\begin{array}{ccc}
	Q^{-1} (A\!-\! E \Gamma^{-1} T_{21} C_q) +(A\!-\! E \Gamma^{-1} T_{21} C_q)^T Q^{-1}+\! \lambda Q^{-1} + Q^{-1} RQ^{-1} & Q^{-1} E \, \Gamma^{-1}  & Q^{-1} G\\
	 \Gamma^{-T} E^T Q^{-1} & -Y^{-1} & 0 \\
	G^T Q^{-1} & 0 & -\lambda I   \end{array}\right] + \\
	&
	\left[\begin{array}{ccc}
	\Sigma C_q & \Lambda & 0 \end{array}\right]^T 
	X^{-1}
	\left[\begin{array}{ccc}
	\Sigma C_q & \Lambda & 0 \end{array}\right] \leq 0,
	\end{split}
	\end{equation*}}
	which then implies that
	{\small
	\begin{equation*}
	\begin{split}
	&\left[\begin{array}{ccc}
	Q^{-1} (A\!-\! E \Gamma^{-1} T_{21} C_q) +(A\!-\! E \Gamma^{-1} T_{21} C_q)^T Q^{-1}+\! \lambda Q^{-1} + Q^{-1} RQ^{-1} & Q^{-1} E \, \Gamma^{-1}  & Q^{-1} G\\
	 \Gamma^{-T} E^T Q^{-1} & 0 & 0 \\
	G^T Q^{-1} & 0 & -\lambda I   \end{array}\right] + \\
	&
	\left[\begin{array}{ccc}
	\Sigma C_q & \Lambda & 0 \\
	0 & I & 0\end{array}\right]^T 
	\left[\begin{array}{cc}
	X^{-1} & 0 \\
	0 & -Y^{-1} \end{array}\right]
	\left[\begin{array}{ccc}
	\Sigma C_q & \Lambda & 0 \\
	0 & I & 0\end{array}\right] \leq 0.
	\end{split}
	\end{equation*}}
	Note that
	$$
	\left[\begin{array}{ccc}
	\Sigma C_q & \Lambda & 0 \\
	0 & I & 0\end{array}\right] =
	\left[\begin{array}{cc}
	\Sigma & \Lambda \\
	0 & I\end{array}\right]
	\left[\begin{array}{ccc}
	 C_q & 0 & 0 \\
	0 & I & 0\end{array}\right].
	$$
	Now post- and pre-multiply the earlier matrix inequality with the following matrix and its transpose
	{\small
	$$
	\left[\begin{array}{ccc}
	I & 0 & 0 \\
	T_{21} C_q  & \Gamma & 0 \\
	0 & 0 & I\end{array}\right]
	$$}
	to obtain
	{\small
	\begin{equation*}
	\begin{split}
	&\left[\begin{array}{ccc}
	Q^{-1} A\!+A^T Q^{-1}+  \! \lambda Q^{-1}\! +\! Q^{-1} R Q^{-1}& Q^{-1} E & Q^{-1} G\\
	  E^T Q^{-1} & 0 & 0 \\
	G^T Q^{-1} & 0 & -\lambda I   \end{array}\right] + \\
	&
	\left[\begin{array}{ccc}
	 C_q & 0 & 0 \\
	T_{21} C_q & \Gamma & 0\end{array}\right]^T
	\left[\begin{array}{cc}
	\Sigma & \Lambda \\
	0 & I\end{array}\right]^T 
	\left[\begin{array}{cc}
	X^{-1} & 0 \\
	0 & -Y^{-1} \end{array}\right]
	\left[\begin{array}{cc}
	\Sigma & \Lambda \\
	0 & I\end{array}\right]
	\left[\begin{array}{ccc}
	 C_q & 0 & 0 \\
	T_{21} C_q & \Gamma & 0\end{array}\right] \leq 0,
	\end{split}
	\end{equation*}}
	where
	{\small
	$$
	\left[\begin{array}{cc}
	\Sigma & \Lambda \\
	0 & I\end{array}\right]
	\left[\begin{array}{ccc}
	 C_q & 0 & 0 \\
	T_{21} C_q & \Gamma & 0\end{array}\right] = \left[\begin{array}{ccc}
	T_{11} C_q & T_{11}D + T_{12} & 0 \\
	T_{21} C_q & T_{21}D +T_{22} & 0\end{array}\right]= T \left[\begin{array}{ccc}
	C_q & D & 0 \\
	0 & I & 0\end{array}\right].
	$$}
	By using Condition \ref{asm:Multip4syn},
	$$
	M=T^T \left[\begin{array}{cc}
	X^{-1} & 0 \\
	0 & -Y^{-1} \end{array}\right] T \in \mathcal{M}
	$$
	 This implies that, for some $M \in \mathcal{M}$, we have
	 {\small
	$$
	\left[\begin{array}{ccc}
	Q^{-1} A\!+A^T Q^{-1}+ \! \lambda Q^{-1} \!+\! Q^{-1} R Q^{-1} & Q^{-1} E & Q^{-1} G\\
	  E^T Q^{-1} & 0 & 0 \\
	G^T Q^{-1} & 0 & -\lambda I   \end{array}\right] + 
	\left[\begin{array}{ccc}
	C_q & D & 0 \\
	0 & I & 0\end{array}\right]^T M \left[\begin{array}{ccc}
	C_q & D & 0 \\
	0 & I & 0\end{array}\right] \leq 0.
	$$}
	Pre- and post-multiplying the above inequality with $[x^T \ p^T \ w^T]$ and its transpose and using $V=x^T Q^{-1}x$, we obtain
	{\small
	$$
	2x^T Q^{-1} (Ax+Ep+Gw) +x^T Q^{-1}RQ^{-1} x  +\lambda(V-\|w\|^2) +
	\left[\begin{array}{c}q \\p\end{array}\right]^T M \left[\begin{array}{c}q \\p\end{array}\right] \leq 0, \quad \mbox{for all} \ \ \, \left[\begin{array}{c} x \\ p \\ w\end{array}\right].
	$$}
	Since $p \in \mathcal{F}(\mathcal{M})$ with $q=Cx+Dp$, by using the S-procedure \cite{LMIbalakbk}, the above inequality implies that the system (\ref{eq:Conesys}) satisfies:
	$
	\dot{V} \leq - x^T Q^{-1}RQ^{-1} x  < 0, \ \ \ \forall \, V \geq \|w\|^2.
	$
\end{proof}
%---
%FOLLOWING part is useful for convergence to origin when disturbances go to zero
%--------------------------------------------------------------------------------
%Now by using Lemma \ref{lm:iae}, we conclude that the set $\Omega(t_0)$ is an invariant and attractive ellipsoid for the system (\ref{eq:Conesys}).

%Now suppose that the signal $w$ is not only bounded but also $\tgi \|w(t)\|=0$.  Consequently for any $\delta >0$ there exists some $t_1>0$ such that $\|w(t) \| < \delta$ for all $t \geq t_1$. Since $\Omega(t_0)$ is invariant and attractive set for the system (\ref{eq:Conesys}), any trajectory of this system is bounded, hence $x(t_1)$ is finite.  Now consider the trajectory of  the system (\ref{eq:Conesys}) for $t \geq t_1$.  By using Lemma \ref{lm:iae} with $t_1$ instead of $t_0$, we can easily conclude that $\Omega(t_1)$ is also an invariant and attractive set for  the system (\ref{eq:Conesys}).  This implies that there exists some time $t_2>t_1$ such that ${\rm d} (x(t),\Omega(t_1)) < \delta$ for all $t > t_2$. Note that $y \in \Omega(t_1)$  implies that $\|y\| \leq \delta/\sqrt{\lambda_{max}(Q)}$.  Hence, for any $t \geq t_2$, $\|x(t)-y\| < \delta$ for all $y \in \Omega(t_1)$, which implies that
%$$
%\|x(t)\|\leq \|x(t)-y\|+\|y\| < \delta \left(1+ \frac{1}{\sqrt{\lambda_{max}(q)}}\right) , \qquad \forall \, t > t_2.
%$$
%Since $\delta >0$ can be chosen arbitrarily small, this implies that $\tgi \|x(t) \|=0$, which completes the proof.
%
The following corollary gives a matrix inequality condition for the quadratic stability of the system (\ref{eq:Conesys}) (when $w=0$), that is, existence of a quadratic Lyapunov function $V = x^T P x$ proving the exponential stability by establishing
\begin{equation} \label{eq:qsCond}
\dot{V} + x^T Q^{-1}RQ^{-1} x  \leq 0 
\end{equation}
for all trajectories of the system (\ref{eq:Conesys}). The proof of the lemma follows  from a straight adaption of the proof of Theorem \ref{lem:UltBoundCone}.
%--
\begin{corol} \label{lem:qs_analysis}
Consider the system given by (\ref{eq:Conesys}) with $w \equiv 0$ and $p \in \mathcal{F}(\mathcal{M})$ where the multiplier set $\mathcal{M}$ satisfies Condition \ref{asm:Multip4syn}.   Suppose that there exist $Q=Q^T>0$, $(X,Y) \in \mathcal{N}$ and  $\lambda >0$ such that the following matrix inequality holds
{\small
\begin{equation} \label{eq:AnalysisLMI_2}
\left[\begin{array}{ccc}
(A\!-\! E \Gamma^{-1} T_{21} C_q)Q +Q(A\!-\! E \Gamma^{-1} T_{21} C_q)^T \! + \! R \! & E \, \Gamma^{-1} Y & Q \, C_q^T \Sigma^T \\
Y\Gamma^{-T}E^T& -Y & Y \Lambda^T \\
\Sigma \, C_q \, Q& \Lambda Y & -X \end{array}\right] \leq 0
\end{equation}}
where $\Gamma, \ \Sigma, \ \Lambda$ are as given in Theorem \ref{lem:UltBoundCone}.
Then the system (\ref{eq:Conesys})  is quadratically stable with a Lyapunov function $V= x^T Q^{-1} x$ and all the trajectories satisfy
\begin{eqnarray} \label{eq:qs->bound}
V(x(t)) &\leq& V(x(t_0)) , \qquad \forall \, t \geq t_0, \\ \label{eq:qs->Opt}
\dot{V}(x) + x^T Q^{-1}RQ^{-1} x &\leq & 0  , \qquad \forall \, x.
\end{eqnarray}
\end{corol}
%---------

%
\bibliographystyle{unsrt}
\bibliography{lmi-bnded}
%====
\end{document}